\documentclass[opre,nonblindrev]{informs3} 

\OneAndAHalfSpacedXI 

\usepackage[utf8]{inputenc}
\usepackage[english]{babel}

\usepackage[toc,page]{appendix}

    \usepackage{natbib}
     \bibpunct[, ]{(}{)}{,}{a}{}{,}%
     \def\bibfont{\small}%
     \usepackage{dsfont}
    \usepackage{amsmath}
     \usepackage{amssymb}
     
     \let\oldemptyset\emptyset
    \let\emptyset\varnothing
    
    \usepackage{enumitem}
     
\TheoremsNumberedThrough     

\EquationsNumberedThrough    

\renewcommand{\theARTICLETOP}{}
\renewcommand{\theMANUSCRIPTNO}{}

\begin{document}




\TITLE{Risk-Averse Equilibrium for Games}

\ARTICLEAUTHORS{%
\AUTHOR{Ali Yekkehkhany}
\AFF{University of Illinois at Urbana-Champaign, \EMAIL{yekkehk2@illinois.edu}
}
\AUTHOR{Timothy Murray}
\AFF{University of Illinois at Urbana-Champaign, \EMAIL{tsmurra2@illinois.edu}
}
\AUTHOR{Rakesh Nagi}
\AFF{University of Illinois at Urbana-Champaign, \EMAIL{nagi@illinois.edu}
}
} 

\ABSTRACT{%
The term \textit{rational} has become synonymous with maximizing expected payoff in the definition of the best response in Nash setting. In this work, we consider stochastic games in which players engage only once, or at most a limited number of times. In such games, it may not be rational for players to maximize their expected payoff as they cannot wait for the Law of Large Numbers to take effect. We instead define a new notion of a risk-averse best response, that results in a risk-averse equilibrium (RAE) in which players choose to play the strategy that maximizes the probability of them being rewarded the most in a single round of the game rather than maximizing the expected received reward, subject to the actions of other players. We prove the risk-averse equilibrium to exist in all finite games and numerically compare its performance to Nash equilibrium in finite-time stochastic games.
}%



\maketitle

%

\section{Introduction}
\label{introduction}
Since the seminal work of \cite{vonNeumann1947}, the term \textit{rational} has become synonymous with expected utility maximization. Whether in game theoretic situations or simply decision-making under uncertainty, the only agent who can be considered rational is the one who attempts to maximize their mean utility, no matter how many trials will likely be necessary for the realized value to resemble the expected value. However, consider an agent faced with multiple options, one of which is an opportunity with maximum expected utility, but it will bankrupt them with high probability if it fails. In the event of failure, consider that the lack of funds will severely limit any future options the agent may have. For such an agent the fact that the opportunity has maximum expected value among the options cannot be the only relevant factor in deciding whether to pursue the opportunity. If the opportunity does not lead to success, the agent will not be able to pursue any later actions, as they will not have the funds necessary to do so. As a result, players should not solely rely on factors such as expected utility and must instead also consider the probability of success for the opportunity. 

This observation applies to almost all stochastic decision-making situations, including competitive situations best modeled through game theory. To see this, consider a market composed of only a few large firms and a smaller firm considering how to compete with large firms or whether to even enter the market. We take as our example the smartphone industry, in which large companies such as Apple, Samsung, Google, LG, Motorola, Amazon, and Microsoft have all competed in recent years. While Apple and Samsung are market leaders at the time of writing, both have undergone expensive setbacks. Apple's iPhone 5 was widely criticized due to issues with the Apple Maps application and Samsung had to recall its Galaxy Note 7 due to its batteries catching fire, costing an estimated 3 billion USD (\cite{swider2016}), in what may have been an attempt to improve on the criticized battery life of their Galaxy S6. Similarly, Google's original Nexus line of phones dropped in popularity to the point where the company went to the expense of creating a new line of Pixel phones rather than continuing the Nexus. Amazon and Microsoft were forced out of the market entirely, with Amazon's Fire Phone lasting just over a year (July 2014 - August 2015) between release and the cessation of production, causing a loss of at least 170 million USD for Amazon's 2014 Q3 alone (\cite{mccallion2014}). Microsoft meanwhile acquired Nokia for 7.2 billion USD in an attempt to become more competitive in the market (\cite{warren2016}), but ceased mobile device production entirely only a few years later.

Despite the cost of the setbacks mentioned above, each of these companies is still valuable with Apple and Microsoft having market caps of over 1 trillion USD at the time of writing and Amazon recently passing that milestone as well. Samsung is worth approximately 300 billion USD at the time of writing, and while they are smaller LG and Motorola are quite valuable as well, worth approximately 14.5 billion and 30 billion USD, respectively. Because of their size, each of these companies was able to take risks to compete with each other which, although expected to end in a positive outcome, resulted in expensive losses. Indeed, Microsoft currently appears to be preparing for another attempt to enter the smartphone market with the Surface Duo. In other words, these companies are still able to compete with each other by making products which maximize their expected values because they are large enough that they can afford to wait for the law of large numbers
to take effect. This allows their competition to be modeled through a traditional game theoretic framework.

In contrast, consider a company with a smaller valuation, say 500 million USD, deciding whether to compete in the smartphone market. If such a company attempted to do so, it would have to commit most if not all of its resources to the attempt. Even if such a strategy has a large positive expected value, it has a large risk of bankrupting the company, as seen with the scale of the losses incurred by Samsung, Amazon, and Microsoft. More generally, firms in markets where the cost of competition is a significant portion of the value of the firm itself must consider more than just maximizing their expected value. A misstep in such a setting means that the firm is out of the market and unable to compete further. This highlights an important facet of competition with random or unknown variables; i.e., it is not just the expected value of a strategy that is important, it is how many times you get to compete.

In this paper, we build a new framework to apply this observation to game theoretic situations. We consider stochastic games drawn from known distributions in which players engage once or for a given finite number of times. Because of the finite number of times that players engage in these games, given the strategies of all other players, expected utility may not be a suitable metric for a player to attempt to maximize.
Instead, we formulate a new definition of a risk-averse best response, where given the strategies of all other agents, an agent chooses to play the strategy that is most likely to have the highest utility in a single realization of the stochastic game. While the mathematical particulars of this definition will be discussed in Section \ref{section_risk_averse_equilibrium}, conceptually it can best be understood through the lens of prospect theory.

In its most basic form, prospect theory states that consumers prefer choices with lower volatility, even when this results in lower expected utility. An excellent example of this is retirement planning where there are many highly volatile assets which in expectation provide a large return on investment, but which also have a high chance of dropping in value due to their volatility. Most individuals try to avoid investing too much in these assets, receiving a lower average return in order to avoid the chance of a significant loss. Similarly, a risk-averse best response as we have loosely defined it so far would possibly limit the expected return of assets in order to maximize the probability of making the most profit.

The rest of this paper is organized as follows. Section \ref{related_work} provides a thorough review of related work relevant to this topic, in particular a more detailed discussion of prospect theory. Section \ref{section_risk_averse_equilibrium} provides the formal mathematical definition of our proposed risk-averse equilibrium, with several subsections detailing topics such as equilibrium properties, computation, and worked-out examples. Section \ref{section_commit_games} considers finite-time commit games and how the risk-averse equilibria shift as the number of times the games are played increases. Section \ref{numerical_results} lays out a comparison between the classical Nash equilibrium and the proposed risk-averse equilibrium through simulation. Finally, Section \ref{conclusion_future} contains our concluding remarks as well as future directions in which to advance this research.

\section{Related Work}
\label{related_work}

Since the seminal work of \cite{vonNeumann1928} and later \cite{nash1950}, expected utility has emerged as the dominant objective value within game theory as each player attempts to maximize his/her expected utility given the actions of other players. This concept was extended naturally into games of incomplete information (Bayesian games) by \cite{harsanyi1967}, as players can still maximize their expected utility given a distribution from which the game will be drawn. These games have received a great deal of attention as they more accurately model real-world situations where not all parameters are known precisely, with later works such as \cite{wiseman2005} addressing how players sequentially refine their equilibria as they learn the distributions and the more recent \cite{mertikopoulos2019} addressing how players learn their payoffs with continuous action sets. Another recent work (\cite{sugaya2019}) considers the more specific question of how firms in a duopoly should play when the payoff distributions are based on the market state, a random variable with possibly unknown distribution.

Despite all the work that has gone into expected utility as the objective value players wish to maximize, it is still questionable whether this is a good assumption. \cite{goeree2003} present an empirical study of risk-aversion in the matching-pennies game, where they observe marked deviations from Nash behavior (expected utility maximization) as the payoffs/costs become larger. This is consistent with the concept of prospect theory based on empirical observations across several experiments in which the subjects deviate from actions which would maximize their expected utility. \cite{kahneman1979prospecttheory:} formulated the idea of prospect theory, which states that consumers are naturally risk-averse when addressing situations with potential gains and naturally risk-seeking when facing situations with potential losses. Prospect theory has since been widely studied, with an extension of the original paper provided in \cite{tversky1992prospectcumulative} to address more general payoff/cost functions. \cite{levy1992politic_prospect} provides a good overview of classical prospect theory, particularly from a political perspective. Unsurprisingly, prospect theory has received a great deal of attention in financial studies (\cite{baele2018,barberis2016}), with \cite{barberis2001} using it for asset pricing. Prospect theory is not without its critics; e.g., \cite{list2004} posits that the results of the studies on prospect theory are due to inexperienced consumers, and designs an experiment to show these behaviors disappear with experience. However, experienced consumers are by definition consumers who engage in similar trials multiple times, which means that for these consumers expected utility \textit{is} an appropriate metric.
As we are explicitly interested in games which will be played at most a small number of times, we do not need to be concerned with this effect.

\section{Risk-Averse Equilibrium}
\label{section_risk_averse_equilibrium}
The problem is formulated in the following subsection.
\subsection{Problem Statement}
\label{problem_statement}
Consider a game that consists of a finite set of $N$ players, $[N] := \{1, 2, \dots, N \}$, where player $i \in [N]$ has a set of possible pure strategies (or actions, used interchangeably) denoted by $S_i$. A pure strategy profile, which is one pure strategy for each player in the game, is denoted by $\boldsymbol{s} = (s_1, s_2, \dots, s_N)$, where $s_i \in S_i$ is the pure strategy of player $i \in [N]$. Hence, $\boldsymbol{S} = S_1 \times S_2 \times \dots \times S_N$ is the set of pure strategy profiles. A pure strategy choice for all players except player $i$ is denoted by $\boldsymbol{s}_{-i}$, i.e. $\boldsymbol{s} = (s_i, \boldsymbol{s}_{-i})$.
The payoff of player $i$ for a pure strategy profile $\boldsymbol{s} \in \boldsymbol{S}$ is denoted by $U_i(\boldsymbol{s})$ (or $U_i(s_i, \boldsymbol{s}_{-i})$), which is a random variable with probability density function (pdf) $f_i(x | \boldsymbol{s})$ and mean $u_i(\boldsymbol{s})$. The payoffs $U_i(\boldsymbol{s})$ for $i \in [N]$ and $\boldsymbol{s} \in \boldsymbol{S}$ are considered to be continuous-type random variables that are independent from each other.
\begin{remark}
The same analysis holds for discrete-type random variables if the analysis is treated with a bit more subtlety as discussed in the end of this section.
\end{remark}

For any set $S_i$, let $\Sigma_i$ be the set of all probability distributions over $S_i$. The Cartesian product of all players' mixed strategy sets, $\boldsymbol{\Sigma} = \Sigma_1 \times \Sigma_2 \times \dots \times \Sigma_N$, is the set of mixed strategy profiles.
Denote a specific mixed strategy of player $i$ by $\sigma_i \in \Sigma_i$, where $\sigma_i(s_i)$ is the probability that player $i$ plays strategy $s_i$.
If the $[N] \setminus i$ players choose to play a mixed strategy $\boldsymbol{\sigma}_{-i}$, the payoff for player $i$ if he/she plays $s_i \in S_i$ is denoted by $\overline{U}_i(s_i, \boldsymbol{\sigma}_{-i})$.
Using the law of total probability, the marginal distribution of $\overline{U}_i(s_i \boldsymbol{\sigma}_{-i})$ has the probability distribution function
\begin{equation}
    \label{mixed_distribution}
    \bar{f}_i(x | (s_i, \boldsymbol{\sigma}_{-i})) = \sum_{\boldsymbol{s}_{-i} \in \boldsymbol{S}_{-i}} \Bigg ( f_i(x | (s_i, \boldsymbol{s}_{-i})) \cdot \boldsymbol{\sigma}(\boldsymbol{s}_{-i}) \Bigg ),
\end{equation}
where $\boldsymbol{\sigma}(\boldsymbol{s}_{-i}) =  \prod_{j \in [N] \setminus i} \sigma_j (s_j)$ and $s_j$ is the corresponding strategy of player $j$ in $\boldsymbol{s}_{-i}$.
Note that for $s_i \neq s_i' \in S_i$, the random variables $\overline{U}_i(s_i, \boldsymbol{\sigma}_{-i})$ and $\overline{U}_i(s_i', \boldsymbol{\sigma}_{-i})$ are not independent of each other in a single play of the game.

\subsection{Risk-Averse Equilibrium}


In a stochastic game where the payoffs are random variables, playing the Nash equilibrium considering the expected payoffs may create a risky situation; e.g., see \cite{yekkehkhany2019risk} and \cite{yekkehkhanycost} and the references therein for examples on multi-armed bandits. The reason is that payoffs with larger expectations may have a larger variance as well. As a result, it may be the case that playing strategies with lower expectations is more probable to have a larger payoff. This concept is mostly helpful when players play the game once, so they do not have the chance to repeat the game and gain a larger cumulative payoff by playing the strategy with the largest expected payoff. As a result, we propose the risk-averse equilibrium in a probabilistic sense rather than in an expectation sense as the Nash equilibrium.
From an individual player's point of view, the best response to a mixed strategy of the rest of players is defined as follows.

\begin{definition}
\label{def_best_response_mixed}
The set of mixed strategy risk-averse best responses of player $i$ to the mixed strategy profile $\boldsymbol{\sigma}_{-i}$ is
the set of all probability distributions over the set
\begin{equation}
    \label{eq_mixed_best_response}
    \argmax_{s_i \in S_i} P \Big ( \overline{U}_i \left ( s_i, \boldsymbol{\sigma}_{-i} \right ) \geq \overline{\boldsymbol{U}}_i \left ( S_i \setminus s_i, \boldsymbol{\sigma}_{-i} \right ) \Big ),
\end{equation}
where what we mean by $U_i \left ( s_i, \boldsymbol{\sigma}_{-i} \right )$ being greater than or equal to $\boldsymbol{U}_i \left ( S_i \setminus s_i, \boldsymbol{\sigma}_{-i} \right )$ when $S_i \setminus s_i \neq \oldemptyset$ is that  $U_i \left ( s_i, \boldsymbol{\sigma}_{-i} \right )$ is greater than or equal to $U_i \left ( s_i', \boldsymbol{\sigma}_{-i} \right )$ for all $s_i' \in S_i \setminus s_i$; otherwise, if $S_i \setminus s_i = \oldemptyset$, player $i$ only has a single option that can be played.
The same randomness on the action of players $[N] \setminus i$ is considered in $\overline{U}_i(s_i, \boldsymbol{\sigma}_{-i})$ for all $s_i \in S_i$, and independent randomness on actions is analyzed in the Appendix.
We denote the risk-averse best response set of player $i$'s strategies, given the other players' mixed strategies $\boldsymbol{\sigma}_{-i}$, by $RB(\boldsymbol{\sigma}_{-i})$, which is in general a set-valued function.
\end{definition}

Given the definition of the risk-averse best response, the risk-averse equilibrium (RAE) is defined as follows.

\begin{definition}
\label{def_mixed_strategy_RAE}
A strategy profile $\boldsymbol{\sigma}^* = (\sigma_1^*, \sigma_2^*, \dots,$ $\sigma_N^*)$ is a risk-averse equilibrium (RAE), if and only if $\sigma_i^* \in RB(\boldsymbol{\sigma}_{-i}^*)$ for all $i \in [N]$.
\end{definition}

The following theorem proves the existence of a mixed strategy risk-averse equilibrium for a game with finite number of players and finite number of strategies per player.

\begin{theorem}
\label{theorem_existence}
For any finite $N$-player game, a risk-averse equilibrium exists.
\end{theorem}

\begin{proof}
{Proof:}Consider the risk-averse best response function $\boldsymbol{RB}: \boldsymbol{\Sigma} \rightarrow \boldsymbol{\Sigma}$ defined as $\boldsymbol{RB}(\boldsymbol{\sigma}) = \big (RB(\boldsymbol{\sigma}_{-1}), RB(\boldsymbol{\sigma}_{-2}),$ $\dots, RB(\boldsymbol{\sigma}_{-N}) \big )$. The existence of a risk-averse equilibrium is equivalent to the existence of a fixed point $\boldsymbol{\sigma}^* \in \boldsymbol{\Sigma}$ such that $\boldsymbol{\sigma}^* \in \boldsymbol{RB}(\boldsymbol{\sigma}^*)$.
Kakutani's Fixed Point Theorem is used to prove the existence of a fixed point for $\boldsymbol{RB}(\boldsymbol{\sigma})$.
In order to use Kakutani's theorem, the four conditions listed below should be satisfied, which are proven as follows.
\begin{enumerate}[leftmargin=*]
    \item $\boldsymbol{\Sigma}$ is a nonempty subset of a finite dimensional Euclidean space, compact, and convex:
    $\boldsymbol{\Sigma}$ is nonempty since it is the Cartesian product of nonempty simplices as each player has at least one feasible pure strategy.
    $\boldsymbol{\Sigma}$ is bounded since each of its elements is between zero and one, and is closed since it is the Cartesian product of simplices, so $\boldsymbol{\Sigma}$ contains all its limit points.
    \item $\boldsymbol{RB(\boldsymbol{\sigma})}$ is nonempty for all $\boldsymbol{\sigma} \in \boldsymbol{\Sigma}$:
    $RB(\boldsymbol{\sigma}_{-i})$ is the set of all probability distributions over the set specified in Equation \eqref{eq_mixed_best_response}, where the mentioned set is nonempty since maximum always exists for finite number of values.
    \item $\boldsymbol{RB(\boldsymbol{\sigma})}$ is a convex set for all $\boldsymbol{\sigma} \in \boldsymbol{\Sigma}$:
    It suffices to prove that $RB(\boldsymbol{\sigma}_{-i})$ is a convex set for all $\boldsymbol{\sigma}_{-i} \in \boldsymbol{\Sigma}_{-i}$. Consider $\sigma_i', \sigma_i'' \in RB(\boldsymbol{\sigma}_{-i})$ and $\lambda \in [0, 1]$. Define the support of $\sigma_i'$ and $\sigma_i''$ as $supp(\sigma_i') = \{s_i \in S_i: \sigma_i'(s_i) > 0\}$ and $supp(\sigma_i'') = \{s_i \in S_i: \sigma_i''(s_i) > 0\}$, respectively.
    From the definition of risk-averse best response in Definition \ref{def_best_response_mixed}, $supp(\sigma_i'), supp(\sigma_i'') \subseteq \argmax_{s_i \in S_i} P \Big ( \overline{U}_i \left ( s_i, \boldsymbol{\sigma}_{-i} \right ) \geq \overline{\boldsymbol{U}}_i \left ( S_i \setminus s_i, \boldsymbol{\sigma}_{-i} \right ) \Big )$. As a result, $supp(\sigma_i') \cup supp(\sigma_i'') \subseteq \argmax_{s_i \in S_i} P \Big ( \overline{U}_i \left ( s_i, \boldsymbol{\sigma}_{-i} \right ) \geq \overline{\boldsymbol{U}}_i \left ( S_i \setminus s_i, \boldsymbol{\sigma}_{-i} \right ) \Big )$, and again due to definition of risk-averse best response, any probability distribution over the set $supp(\sigma_i') \cup supp(\sigma_i'')$ is also a best response to $\boldsymbol{\sigma}_{-i}$. The mixed strategy $\lambda \sigma_i' + (1 - \lambda) \sigma_i''$ is obviously a valid probability distribution over the set $supp(\sigma_i') \cup supp(\sigma_i'')$, so $\lambda \sigma_i' + (1 - \lambda) \sigma_i'' \in RB(\boldsymbol{\sigma}_{-i})$ that completes the proof for convexity of the set $RB(\boldsymbol{\sigma}_{-i})$.
    \item $\boldsymbol{RB(\boldsymbol{\sigma})}$ has a closed graph:
    $\boldsymbol{RB}(\boldsymbol{\sigma})$ has a closed graph if for any sequence $\{\boldsymbol{\sigma}^n, \widehat{\boldsymbol{\sigma}}^n\} \rightarrow \{\boldsymbol{\sigma}, \widehat{\boldsymbol{\sigma}}\}$ with $\widehat{\boldsymbol{\sigma}}^n \in \boldsymbol{RB}(\boldsymbol{\sigma}^n)$ for all $n \in \mathbb{N}$, we have $\widehat{\boldsymbol{\sigma}} \in \boldsymbol{RB}(\boldsymbol{\sigma})$. The fact that $\boldsymbol{RB}(\boldsymbol{\sigma})$ has a closed graph is proved by contradiction. Consider that $\boldsymbol{RB}(\boldsymbol{\sigma})$ does not have a closed graph.
    Then, there exists a sequence $\{\boldsymbol{\sigma}^n, \widehat{\boldsymbol{\sigma}}^n\} \rightarrow \{\boldsymbol{\sigma}, \widehat{\boldsymbol{\sigma}}\}$ with $\widehat{\boldsymbol{\sigma}}^n \in \boldsymbol{RB}(\boldsymbol{\sigma}^n)$ for all $n \in \mathbb{N}$, but $\widehat{\boldsymbol{\sigma}} \notin \boldsymbol{RB}(\boldsymbol{\sigma})$. This means there exists some $i \in [N]$ such that $\widehat{\sigma}_i \notin RB(\boldsymbol{\sigma}_{-i})$. As a result, due to the definition of risk-averse best response in Definition \ref{def_best_response_mixed}, there exists $\widehat{s}_i \in supp(\widehat{\sigma}_i)$, $s_i' \in S_i$, where $s_i'$ can be any of the strategies in the set $supp(RB(\boldsymbol{\sigma}_{-i}))$, and some $\epsilon > 0$ such that
    \begin{equation}
        \label{eq_contradiction_1}
        P \Big ( \overline{U}_i \left ( s_i', \boldsymbol{\sigma}_{-i} \right ) \geq \overline{\boldsymbol{U}}_i \left ( S_i \setminus s_i', \boldsymbol{\sigma}_{-i} \right ) \Big ) > P \Big ( \overline{U}_i \left ( \widehat{s}_i, \boldsymbol{\sigma}_{-i} \right ) \geq \overline{\boldsymbol{U}}_i \left ( S_i \setminus \widehat{s}_i, \boldsymbol{\sigma}_{-i} \right ) \Big ) + 3 \epsilon.
    \end{equation}
    Given that payoffs are continuous random variables and $\boldsymbol{\sigma}_{-i}^n \rightarrow \boldsymbol{\sigma}_{-i}$, for a sufficiently large $n$ we have
    \begin{equation}
        \label{eq_contradiction_2}
        P \Big ( \overline{U}_i \left ( s_i', \boldsymbol{\sigma}_{-i}^n \right ) \geq \overline{\boldsymbol{U}}_i \left ( S_i \setminus s_i', \boldsymbol{\sigma}_{-i}^n \right ) \Big ) > P \Big ( \overline{U}_i \left ( s_i', \boldsymbol{\sigma}_{-i} \right ) \geq \overline{\boldsymbol{U}}_i \left ( S_i \setminus s_i', \boldsymbol{\sigma}_{-i} \right ) \Big ) - \epsilon.
    \end{equation}
    By combining Equations \eqref{eq_contradiction_1} and \eqref{eq_contradiction_2}, for a sufficiently large $n$ we have
    \begin{equation}
        \label{eq_contradiction_3}
        P \Big ( \overline{U}_i \left ( s_i', \boldsymbol{\sigma}_{-i}^n \right ) \geq \overline{\boldsymbol{U}}_i \left ( S_i \setminus s_i', \boldsymbol{\sigma}_{-i}^n \right ) \Big ) > P \Big ( \overline{U}_i \left ( \widehat{s}_i, \boldsymbol{\sigma}_{-i} \right ) \geq \overline{\boldsymbol{U}}_i \left ( S_i \setminus \widehat{s}_i, \boldsymbol{\sigma}_{-i} \right ) \Big ) + 2 \epsilon.
    \end{equation}
    Due to the same reasoning as for Equation \eqref{eq_contradiction_2}, for a sufficiently large $n$ we have
    \begin{equation}
        \label{eq_contradiction_4}
        P \Big ( \overline{U}_i \left ( \widehat{s}_i^n, \boldsymbol{\sigma}_{-i}^n \right ) \geq \overline{\boldsymbol{U}}_i \left ( S_i \setminus \widehat{s}_i^n, \boldsymbol{\sigma}_{-i}^n \right ) \Big ) < P \Big ( \overline{U}_i \left ( \widehat{s}_i, \boldsymbol{\sigma}_{-i} \right ) \geq \overline{\boldsymbol{U}}_i \left ( S_i \setminus \widehat{s}_i, \boldsymbol{\sigma}_{-i} \right ) \Big ) + \epsilon,
    \end{equation}
    where $\widehat{s}_i^n \in supp(RB(\boldsymbol{\sigma}_{-i}^n))$. Combining Equations \eqref{eq_contradiction_3} and \eqref{eq_contradiction_4}, for a sufficiently large $n$ we have
    \begin{equation}
        \label{eq_contradiction_5}
        P \Big ( \overline{U}_i \left ( s_i', \boldsymbol{\sigma}_{-i}^n \right ) \geq \overline{\boldsymbol{U}}_i \left ( S_i \setminus s_i', \boldsymbol{\sigma}_{-i}^n \right ) \Big ) > P \Big ( \overline{U}_i \left ( \widehat{s}_i^n, \boldsymbol{\sigma}_{-i}^n \right ) \geq \overline{\boldsymbol{U}}_i \left ( S_i \setminus \widehat{s}_i^n, \boldsymbol{\sigma}_{-i}^n \right ) \Big ) +  \epsilon.
    \end{equation}
    However, Equation \eqref{eq_contradiction_5} contradicts the fact that $\widehat{s}_i^n \in supp(RB(\boldsymbol{\sigma}_{-i}^n))$.
\end{enumerate}
The above four properties of the risk-averse best response function $\boldsymbol{RB}(\boldsymbol{\sigma})$ fulfil the conditions for Kakutani's Fixed Point Theorem. This means that for a finite $N$-player game, there always exists $\boldsymbol{\sigma}^* \in \boldsymbol{\Sigma}$ such that $\boldsymbol{\sigma}^* \in \boldsymbol{RB}(\boldsymbol{\sigma}^*)$, where by definition $\boldsymbol{\sigma}^*$ is a mixed strategy risk-averse equilibrium. \Halmos
\end{proof}

\subsection{Pure Strategy Risk-Averse Equilibrium}
\label{pure_Strategy_RAE}
The pure strategy risk-averse best response is defined in the following as a specific case of the risk-averse best response defined in Definition \ref{def_best_response_mixed}.

\begin{definition}
\label{def_best_response}
Pure strategy $\widehat{s}_i$ of player $i$ is a risk-averse best response (RB) to the pure strategy $\boldsymbol{s}_{-i}$ of the other players if
\begin{equation}
    \label{pure_best_response}
    \left\{
    \begin{array}{ll}
        \widehat{s}_i \in \argmax_{s_i \in S_i} P \Big ( U_i \left ( s_i, \boldsymbol{s}_{-i} \right ) \geq \boldsymbol{U}_i \left ( S_i \setminus s_i, \boldsymbol{s}_{-i} \right ) \Big ), \ \ \ \text{ if } S_i \setminus s_i \neq \oldemptyset,
        \\
        \widehat{s}_i = s_i, \ \ \  \text{ if } S_i \setminus s_i = \oldemptyset,
    \end{array}
    \right.
\end{equation}
where what we mean by $U_i \left ( s_i, \boldsymbol{s}_{-i} \right )$ being greater than or equal to $\boldsymbol{U}_i \left ( S_i \setminus s_i, \boldsymbol{s}_{-i} \right )$ is that $U_i \left ( s_i, \boldsymbol{s}_{-i} \right )$ is greater than or equal to $U_i \left ( s_i', \boldsymbol{s}_{-i} \right )$ for all $s_i' \in S_i \setminus s_i$. We denote the risk-averse best response set of player $i$, given the other players' pure strategies $\boldsymbol{s}_{-i}$, by $RB(\boldsymbol{s}_{-i})$ (overloading notation, $RB(.)$ is used for both pure and mixed strategy risk-averse best response).
\end{definition}

Given the definition of the pure strategy risk-averse best response, the pure strategy risk-averse equilibrium (RAE), which does not necessarily exist, is defined below.

\begin{definition}
\label{def_pure_strategy_RAE}
A pure strategy profile $\boldsymbol{s}^* = (s_1^*, s_2^*, \dots,$ $s_N^*)$ is a pure strategy risk-averse equilibrium (RAE), if and only if $s_i^* \in RB(\boldsymbol{s}_{-i}^*)$ for all $i \in [N]$.
\end{definition}

\subsection{Strict Dominance and Iterated Elimination of Strictly Dominated Strategies}

Probably the most basic solution concept for a game is the dominant strategy equilibrium.
In the following definition, the strict dominance is described.

\begin{definition}
\label{def_strict_dominance}
A pure strategy $s_i \in S_i$ of player $i$ strictly dominates a second pure strategy $s_i' \in S_i$ of the player if
\begin{equation}
    \label{eq_strict_dominance}
    P \Big ( U_i \left ( s_i, \boldsymbol{s}_{-i} \right ) \geq \boldsymbol{U}_i \left ( S_i \setminus s_i, \boldsymbol{s}_{-i} \right ) \Big ) > P \Big ( U_i \left ( s_i', \boldsymbol{s}_{-i} \right ) \geq \boldsymbol{U}_i \left ( S_i \setminus s_i', \boldsymbol{s}_{-i} \right ) \Big ), \forall \boldsymbol{s}_{-i} \in \boldsymbol{S}_{-i}.
\end{equation}
\end{definition}

A strictly dominated strategy cannot be the risk-averse best response to any mixed strategy profile of other players due to the following reason. Consider that $s_i' \in S_i$ is strictly dominated by $s_i \in S_i$ for player $i$ as is stated in Definition \ref{def_strict_dominance}. Then, for any $\boldsymbol{\sigma}_{-i} \in \boldsymbol{\Sigma}_{-i}$, we have

\begin{equation}
    \label{eq_proof_strict_dominance}
    \begin{aligned}
        & P \Big ( \overline{U}_i(s_i, \boldsymbol{\sigma}_{-i}) \geq \overline{U}_i(S_i \setminus s_i, \boldsymbol{\sigma}_{-i}) \Big ) \\
        \overset{(a)}{=} & \sum_{\boldsymbol{s}_{-i} \in \boldsymbol{S}_{-i}} \Bigg ( P \Big ( U_i(s_i, \boldsymbol{s}_{-i}) \geq \boldsymbol{U}_i(S_i \setminus s_i, \boldsymbol{s}_{-i}) \Big ) \cdot \boldsymbol{\sigma}(\boldsymbol{s}_{-i}) \Bigg ) \\
        \overset{(b)}{>} & \sum_{\boldsymbol{s}_{-i} \in \boldsymbol{S}_{-i}} \Bigg ( P \Big ( U_i(s_i', \boldsymbol{s}_{-i}) \geq \boldsymbol{U}_i(S_i \setminus s_i', \boldsymbol{s}_{-i}) \Big ) \cdot \boldsymbol{\sigma}(\boldsymbol{s}_{-i}) \Bigg ) \\
        = & \ P \Big ( \overline{U}_i(s_i', \boldsymbol{\sigma}_{-i}) \geq \overline{U}_i(S_i \setminus s_i', \boldsymbol{\sigma}_{-i}) \Big ),
    \end{aligned}
\end{equation}
where $(a)$ is followed by using the law of total probability, $\boldsymbol{\sigma}(\boldsymbol{s}_{-i}) =  \prod_{j \in [N] \setminus i} \sigma_j (s_j)$ and $s_j$ is the corresponding strategy of player $j$ in $\boldsymbol{s}_{-i}$, and
$(b)$ is true by the assumption that the pure strategy $s_i'$ is strictly dominated by the pure strategy $s_i$ and using Equation \eqref{eq_strict_dominance} in Definition \ref{def_strict_dominance} on strict dominance.
By Equation \eqref{eq_proof_strict_dominance} and Equation \eqref{eq_mixed_best_response} in Definition \ref{def_best_response_mixed} on the best response to a mixed strategy profile of other players, a strictly dominated pure strategy can never be a best response to any mixed strategy profile of other players. As a result, a strictly dominated pure strategy can be removed from the set of strategies of a player and iterated elimination of strictly dominated strategies can be applied to a game under the risk-averse framework.

\subsection{Finding the Risk-Averse Equilibrium}
The mixed strategy risk-averse equilibrium of a game can be found by choosing players' mixed strategy profiles in such a way that a player cannot strategize against other players. In other words, under a mixed strategy risk-averse equilibrium, all players are indifferent to their mixed strategies, so they use a mixed strategy to make other players indifferent as well. If all players are indifferent to their mixed risk-averse strategies, then no player has an incentive to change strategies, so we end up with a mixed strategy risk-averse equilibrium. Formally speaking, a risk-averse mixed strategy is characterized by $\sigma_i(s_i)$ for all $i \in [N]$ and for all $s_i \in S_i$, so there are $\sum_{i \in [N]} |S_i|$ parameters that should be found.
Letting the mixed strategy profile for players $[N] \setminus i$ be $\boldsymbol{\sigma}_{-i} \in \boldsymbol{\Sigma}_{-i}$, then in order for player $i$ to be indifferent to his/her set of strategies among a subset $S_i'\subseteq S_i$, we need to have
\begin{equation*}
    P \Big ( \overline{U}_i \left ( s_i', \boldsymbol{\sigma}_{-i} \right ) \geq \overline{\boldsymbol{U}}_i \left ( S_i \setminus s_i', \boldsymbol{\sigma}_{-i} \right ) \Big ) \geq P \Big ( \overline{U}_i \left ( s_i, \boldsymbol{\sigma}_{-i} \right ) \geq \overline{\boldsymbol{U}}_i \left ( S_i \setminus s_i, \boldsymbol{\sigma}_{-i} \right ) \Big ), \ \forall s_i\in S_i, s_i' \in S_i'
\end{equation*}
The above equations reveal $|S_i| - 1$ independent equations for each player $i$, so in total $\sum_{i \in [N]} |S_i| - N$ equations are derived. The remaining $N$ equations are provided by the fact that the mixed strategy of each player adds to one for their set of strategies. As a result, if there is a mixed strategy risk-averse equilibrium for which only a subset $\boldsymbol{S}'=\{S_1',S_2',...,S_N'\}$ of the pure strategies, denoted as the \textit{support} of the equilibrium, are played with non-zero probability, this equilibrium is a solution of the following set of equations for $\boldsymbol{\sigma} \in \boldsymbol{\Sigma}$:
\begin{equation}
\label{eq_find_mixed_strategy_RAE}
    \left\{
    \begin{array}{ll}
        P \Big ( \overline{U}_i \left ( s_i', \boldsymbol{\sigma}_{-i} \right ) \geq \overline{\boldsymbol{U}}_i \left ( S_i \setminus s_i', \boldsymbol{\sigma}_{-i} \right ) \Big ) \geq P \Big ( \overline{U}_i \left ( s_i, \boldsymbol{\sigma}_{-i} \right ) \geq \overline{\boldsymbol{U}}_i \left ( S_i \setminus s_i, \boldsymbol{\sigma}_{-i} \right ) \Big ), \ \forall s_i\in S_i, s_i' \in S_i', \forall i \in [N],\\
        \\
        \sum_{s_i \in S_i} \sigma_i(s_i) = 1, \forall i \in [N],\\
        \\
        \sigma_i(s_i) = 0, \forall s_i \notin S_i', \forall i \in [N].
    \end{array}
    \right.
\end{equation}
Any solution to Equation set \eqref{eq_find_mixed_strategy_RAE} is a risk-averse equilibrium, so we can check if an equilibrium exists for any support $\boldsymbol{S}'\subseteq \boldsymbol{\Sigma}$.

Note that as is stated in Equation \eqref{eq_proof_strict_dominance}, we have the following by using the law of total probability:
\begin{equation}
\label{prob_matrix}
\begin{aligned}
    & P \Big ( \overline{U}_i \left ( s_i, \boldsymbol{\sigma}_{-i} \right ) \geq \overline{\boldsymbol{U}}_i \left ( S_i \setminus s_i, \boldsymbol{\sigma}_{-i} \right ) \Big ) \\
    = & \sum_{s_{-i} \in \boldsymbol{S}_{-i}} \Bigg (  \boldsymbol{\sigma}(\boldsymbol{s}_{-i}) \cdot P \Big (U_i(s_i, \boldsymbol{s}_{-i}) \geq \boldsymbol{U}_i(S_i \setminus s_i, \boldsymbol{s}_{-i}) \Big )  \Bigg ),
\end{aligned}
\end{equation}
where $\boldsymbol{\sigma}(\boldsymbol{s}_{-i}) =  \prod_{j \in [N] \setminus i} \sigma_j (s_j)$ and $s_j$ is the corresponding strategy of player $j$ in $\boldsymbol{s}_{-i}$. Hence, Equation \eqref{prob_matrix} is polynomial of order $N - 1$ in terms of $\sigma(s_i)$ for $s_i \in S_i$ and $i \in [N]$.
We can define a \textit{risk-averse probability tensor} of dimension $|S_1| \times |S_2| \times \dots \times |S_N|$, where the $i$-th dimension has all pure strategies $s_i \in S_i$ and each element of the tensor is an $N$ dimensional vector defined in the following. The $i$-th element of the $N$ dimensional vector corresponding to the pure strategy profile $(s_i, \boldsymbol{s}_{-i})$ is defined as
\begin{equation}
    \label{eq_risk_averse_probability_matrix}
    p_i(s_i, \boldsymbol{s}_{-i}) = P \Big ( U_i \left ( s_i, \boldsymbol{s}_{-i} \right ) \geq \boldsymbol{U}_i \left ( S_i \setminus s_i, \boldsymbol{s}_{-i} \right ) \Big ).
\end{equation}
As a result, an equivalent approach for finding the risk-averse equilibrium is to find the Nash equilibrium of the risk-averse probability tensor, as any such Nash equilibrium must maximize the probability of playing a utility-maximizing response to $\boldsymbol{\sigma}_{-i}$ for each player $i$.
In the following two subsections, two illustrative examples are provided to make the concept of the risk-averse equilibrium clear.



\subsection{Illustrative Example 1}
The following example is presented to shed light on the notion of pure strategy risk-averse equilibrium.
\label{illustrative_example}
\begin{example}
\label{example1}
Consider a game between two players where each player has two pure strategies, $S_1 = \{U, D\}$ and $S_2 = \{L, R\}$, with independent payoff distributions specified as
\begin{enumerate}[label=(\roman*)]
    \item $U_1(U, L)$ and $U_2(U, L)$ are independent and have the same pdf as $f_4(u) = \alpha \Big ( 3e^{-20(u - 2)^2} \cdot \mathds{1}\{ \frac32 \leq u \leq \frac52 \} + 2e^{-20(u - 7)^2} \cdot \mathds{1}\{ \frac{13}{2} \leq u \leq \frac{15}{2} \} \Big )$,
    \item $U_1(U, R)$ and $U_2(U, R)$ are independent and have the same pdf as $f_3(u) = \beta e^{-20(u - 3)^2} \cdot \mathds{1}\{ \frac52 \leq u \leq \frac72 \}$,
    \item $U_1(D, L)$ and $U_2(D, L)$ are independent and have the same pdf as $\widehat{f}_3(u) = \gamma \Big ( 3e^{-20(u - 1)^2} \cdot \mathds{1}\{ \frac12 \leq u \leq \frac32 \} + 2e^{-20(u - 6)^2} \cdot \mathds{1}\{ \frac{11}{2} \leq u \leq \frac{13}{2} \} \Big )$,
    \item $U_1(D, R)$ and $U_2(D, R)$ are independent and have the same pdf as $f_5(u) = \delta \Big ( 7e^{-20(u - 2)^2} \cdot \mathds{1}\{ \frac32 \leq u \leq \frac52 \} + 3e^{-20(u - 12)^2} \cdot \mathds{1}\{ \frac{23}{2} \leq u \leq \frac{25}{2} \} \Big )$,
\end{enumerate}
where $\alpha, \beta, \gamma$, and $\delta$ are constants for which each of the corresponding distributions integrate to one and $\mathds{1}\{.\}$ is the indicator function.

\begin{figure}[t]
\centering
\includegraphics[width=0.75\textwidth]{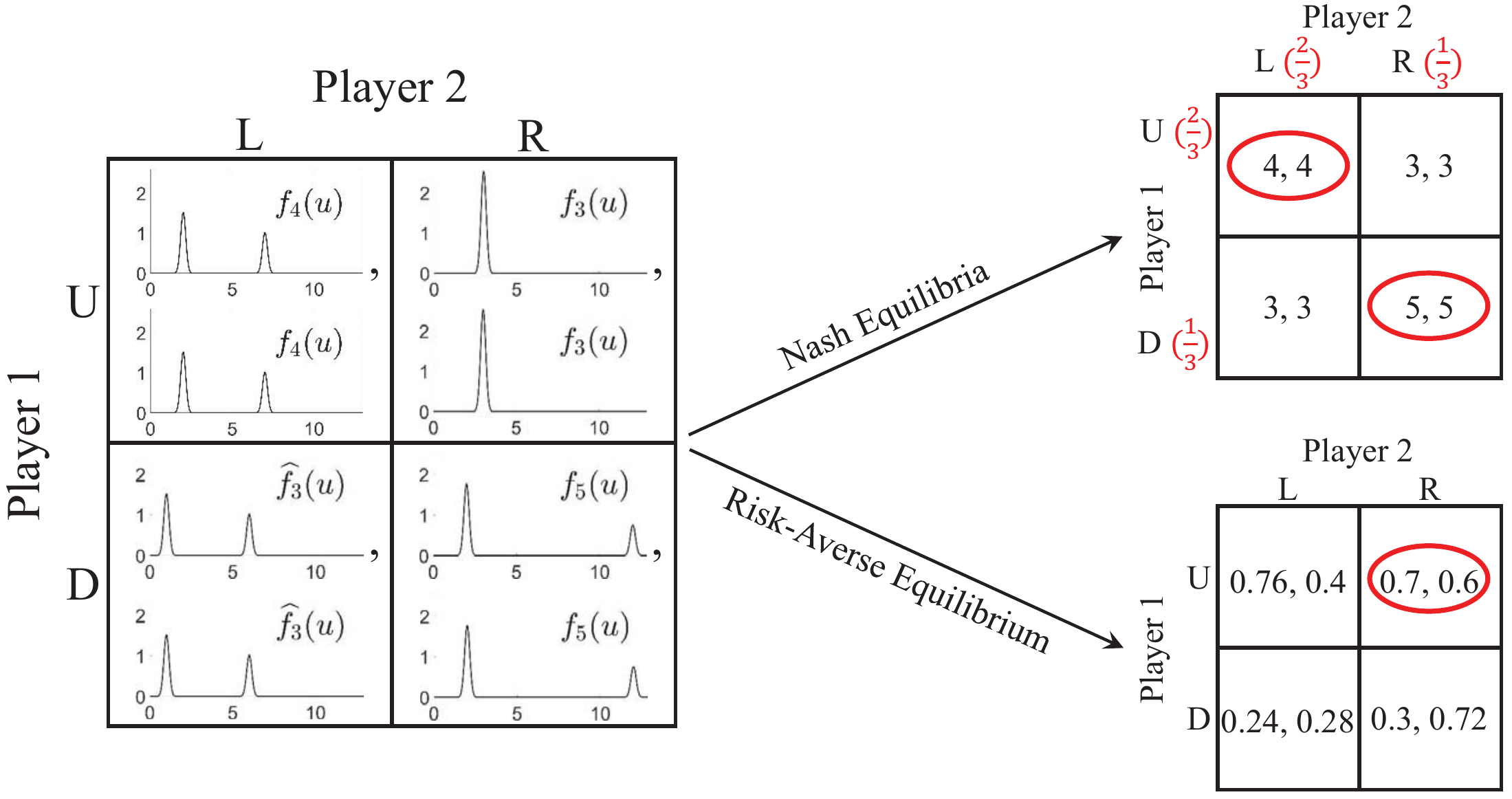}
\caption{The payoff matrix of Example \ref{example1}. The pure and mixed strategy Nash equilibria are shown on the top-right and the pure strategy risk-averse equilibrium is shown on the bottom-right.}
\label{figure_example1}
\end{figure}

\end{example}

The above example is depicted in Figure \ref{figure_example1}.
Considering the expected payoffs in Example \ref{example1}, $\mathds{E} \left [ U_1(U, L) \right ] = \mathds{E} \left [ U_2(U, L) \right ] = 4$, $\mathds{E} \left [ U_1(U, R) \right ] = \mathds{E} \left [ U_2(U, R) \right ] = \mathds{E} \left [ U_1(D, L) \right ] = \mathds{E} \left [ U_2(D, L) \right ] = 3$, and $\mathds{E} \left [ U_1(D, R) \right ] = \mathds{E} \left [ U_2(D, R) \right ] = 5$, the pure Nash equilibria of the game are $(U, L)$ and $(D, R)$, and the mixed Nash equilibrium is that the first player selects $U$ with probability two-thirds and selects $D$ otherwise and the second player selects $L$ with probability two-thirds and selects $R$ otherwise.
On the other hand, by using the payoff density functions we have $P \big (U_1(U, L) \geq U_1(D, L) \big ) = 0.76, P \big (U_1(U, R) \geq U_1(D, R) \big ) = 0.7, P \big (U_2(U, L) \geq U_2(U, R) \big ) = 0.4,$ and $P \big (U_2(D, L) \geq U_2(D, R) \big ) = 0.28$, which are used to form the risk-averse probability bi-matrix of the game derived based on Equation \eqref{eq_risk_averse_probability_matrix}. The risk-averse probability matrix is depicted in Figure \ref{figure_example1}. According to Definition \ref{def_pure_strategy_RAE}, $(U, R)$ is a pure strategy risk-averse equilibrium that is different from the Nash equilibria of the game.
Taking a close look at the payoff distributions, $(U, R)$ is less risky than $(U, L)$ and $(D, R)$ in a single round of the game.

\subsection{Illustrative Example 2}
\label{illustrative_example2}
In this subsection, the mixed strategy risk-averse equilibrium of a two-player game proposed in the following example is computed.
\begin{example}
\label{example2}
Consider a game between two players where each player has two pure strategies, $S_1 = \{U, D\}$ and $S_2 = \{L, R\}$, with independent payoff distributions specified as
\begin{enumerate}[label=(\roman*)]
    \item $U_1(U, L)$ and $U_2(U, L)$ are independent and have the same pdf as $f_4(u) = \alpha \Big ( 3e^{-20(u - 2)^2} \cdot \mathds{1}\{ \frac32 \leq u \leq \frac52 \} + 2e^{-20(u - 7)^2} \cdot \mathds{1}\{ \frac{13}{2} \leq u \leq \frac{15}{2} \} \Big )$,
    \item $U_1(U, R), U_2(U, R), U_1(D, L)$, and $U_2(D, L)$ are independent and have the same pdf as $f_3(u) = \beta e^{-20(u - 3)^2} \cdot \mathds{1}\{ \frac52 \leq u \leq \frac72 \}$,
    \item $U_1(D, R)$ and $U_2(D, R)$ are independent and have the same pdf as $f_1(u) = \gamma e^{-20(u - 1)^2} \cdot \mathds{1}\{ \frac12 \leq u \leq \frac32 \}$,
\end{enumerate}
where $\alpha, \beta,$ and $\gamma$ are constants for which each of the corresponding distributions integrate to one.
\end{example}

\begin{figure}[t]
\centering
\includegraphics[width=0.75\textwidth]{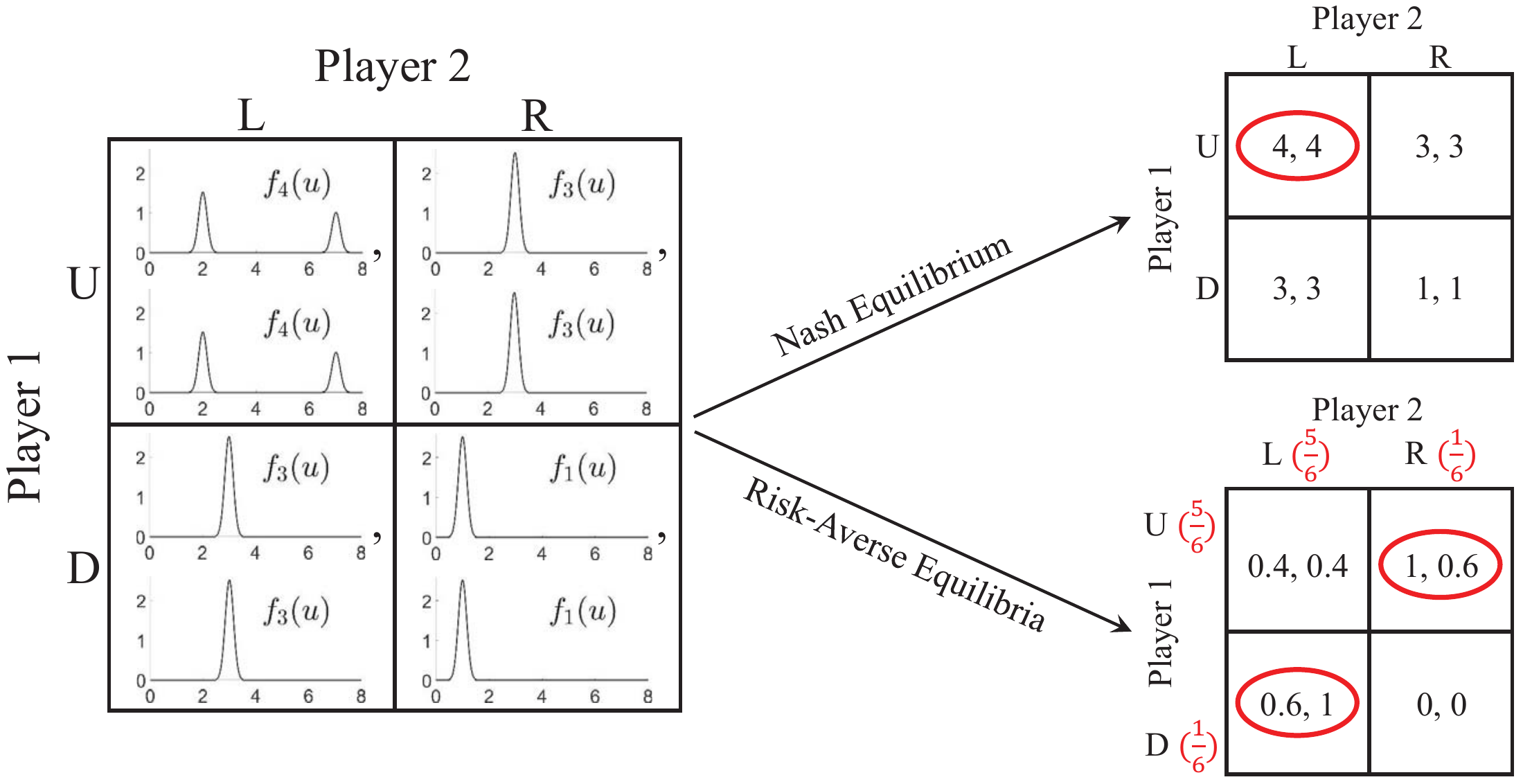}
\caption{The payoff matrix of Example \ref{example2}. The pure strategy Nash equilibrium is shown on the top-right and the pure and mixed strategy risk-averse equilibria are shown on the bottom-right.}
\label{figure_example2}
\end{figure}

The above example is depicted in Figure \ref{figure_example2}.
Considering the expected payoffs in Example \ref{example2}, $\mathds{E} \left [ U_1(U, L) \right ] = \mathds{E} \left [ U_2(U, L) \right ] = 4$, $\mathds{E} \left [ U_1(U, R) \right ] = \mathds{E} \left [ U_2(U, R) \right ] = \mathds{E} \left [ U_1(D, L) \right ] = \mathds{E} \left [ U_2(D, L) \right ] = 3$, and $\mathds{E} \left [ U_1(D, R) \right ] = \mathds{E} \left [ U_2(D, R) \right ] = 1$, the pure Nash equilibrium of the game is $(U, L)$ as depicted in Figure \ref{figure_example2} with no mixed strategy Nash equilibria.
On the other hand, by using the payoff density functions we have $P \big (U_1(U, L) \geq U_1(D, L) \big ) = 0.4, P \big (U_1(U, R) \geq U_1(D, R) \big ) = 1, P \big (U_2(U, L) \geq U_2(U, R) \big ) = 0.4,$ and $P \big (U_2(D, L) \geq U_2(D, R) \big ) = 1$, which are used to form the risk-averse probability bi-matrix of the game derived based on Equation \eqref{eq_risk_averse_probability_matrix}. The risk-averse probability matrix is depicted in Figure \ref{figure_example2}. According to Definition \ref{def_pure_strategy_RAE}, $(U, R)$ and $(D, L)$ are the pure strategy risk-averse equilibria.
In order to find the mixed strategy risk-averse equilibrium, consider that the first player selects $U$ with probability $\sigma_U$ and selects $D$ otherwise. Given the first player's mixed strategy $(\sigma_U, 1 - \sigma_U)$, with a little misuse of notation, denote the random variables denoting the second player's payoffs by selecting $L$ or $R$ with $L$ and $R$, respectively.
The second player is indifferent between selecting $L$ and $R$ if $P(L \geq R) = P(R \geq L)$. Since payoffs are continuous random variables, $P(R \geq L) = 1 - P(L \geq R)$; as a result, the second player is indifferent between the strategies if $P(L \geq R) = 0.5$.
By using the law of total probability and independence of payoff distributions, $P(L \geq R)$ can be computed as
\begin{equation}
    \begin{aligned}
        P(L \geq R) & = \sigma_U \cdot P \Big ( U_2(U, L) \geq U_2(U, R) \Big | \text{Player $1$ plays $U$} \Big ) \\
        & \ \ \ + (1 - \sigma_U) \cdot P \Big ( U_2(D, L) \geq U_2(D, R) \Big | \text{Player $1$ plays $D$} \Big ) \\
        & = \sigma_U \cdot \int_{-\infty}^\infty \int_v^\infty f_4(u) \cdot f_3(v) \ dudv + (1 - \sigma_U) \cdot \int_\infty^\infty \int_v^\infty f_3(u) \cdot f_1(v) \ dudv \\
        & = \frac25\sigma_U + (1 - \sigma_U) = 1 - \frac35\sigma_U.
    \end{aligned}
\end{equation}
Letting $P(L \geq R) = 0.5$, then $\sigma_U = \frac56$, which determines the mixed strategy risk-averse equilibrium. As a result, due to symmetry, $\big (\sigma_1(U), \sigma_1(D) \big ) = (\frac56, \frac16)$ and $\big (\sigma_2(L), \sigma_2(R) \big ) = (\frac56, \frac16)$ form the mixed strategy risk-averse equilibrium of the game in Example \ref{example2}.

It is easy to verify that the game proposed in Example \ref{example1} does not have any mixed strategy risk-averse equilibria. The game in Example \ref{example1} has both pure and mixed strategy Nash equilibria, but it only has pure strategy risk-averse equilibrium. On the other hand, the game in Example \ref{example2} only has pure strategy Nash equilibrium, but it has both pure and mixed risk-averse equilibria. As can be seen, the distributions of payoffs can have a significant impact on the behavior of players if they take risk into account when taking their decisions.

\begin{remark}
As mentioned earlier in this section, the analysis for risk-averse equilibrium holds for discrete-time random variables as well. For example, consider random variables $X, Y$, and $Z$ with distributions
\[
\begin{aligned}
    & P(X = 1) = 0.8, \ P(X = 2) = 0.2, \\
    & P(Y = 1) = 1, \\
    & P(Z = 1) = 0.5, \ P(Z = 2) = 0.5.
\end{aligned}
\]
Denote the observations of the three random variables by triple $(X, Y, Z)$ and let $\{X \geq (Y, Z) \}$ be the event that $X$ is greater than or equal to both $Y$ and $Z$. Then
\[
\begin{aligned}
    & P \big (X \geq (Y, Z) \big ) = P \big ( \{ (1, 1, 1), (2, 1, 1), (2, 1, 2) \} \big ) = 0.4 + 0.1 + 0.1 = 0.6, \\
    & P \big ( Y \geq (X, Z) \big ) = P \big ( (1, 1, 1) \big ) = 0.4, \\
    & P \big ( Z \geq (X, Y) \big ) = P \big ( \{ (1, 1, 1), (1, 1, 2), (2, 1, 2) \} \big ) = 0.4 + 0.4 + 0.1 = 0.9.
\end{aligned}
\]
As can be seen, $P \big (X \geq (Y, Z) \big ) + P \big ( Y \geq (X, Z) \big ) + P \big ( Z \geq (X, Y) \big ) = 1.9 > 1$. In order to resolve this issue, we can break ties uniformly at random as
\[
\begin{aligned}
    & P \big (X \geq (Y, Z) \big ) = \frac13 \times 0.4 + 0.1 + \frac12 \times 0.1 = \frac{17}{60}, \\
    & P \big ( Y \geq (X, Z) \big ) = \frac13 \times 0.4 = \frac{2}{15}, \\
    & P \big ( Z \geq (X, Y) \big ) = \frac13 \times 0.4 + 0.4 + \frac12 \times 0.1 = \frac{35}{60},
\end{aligned}
\]
which results in $P \big (X \geq (Y, Z) \big ) + P \big ( Y \geq (X, Z) \big ) + P \big ( Z \geq (X, Y) \big ) = 1$.
\end{remark}

\section{Finite-Time Commit Games}
The risk-averse framework discussed in Section \ref{section_risk_averse_equilibrium} provides risk-averse players with pure or mixed strategies such that given the other players' strategies, risk-averse equilibrium maximizes the probability that a player is rewarded the most in a single round of the game rather than maximizing the expected received reward. 
On the other hand, for infinite rounds of playing the game, given the other players' strategies, selecting the strategy that maximizes the expected reward guarantees maximum cumulative reward. However, the rewards may not be satisfying for a risk-averse player in each and every round of playing the game.
As a result, risk-averse players may even choose to play the risk-averse equilibrium in infinite (or finite) rounds of games to have more or less balanced rewards in all rounds of the game rather than have maximum cumulative reward in the end.
Despite this fact, we present a slightly different approach for finite-time games that aims to maximize not the expected cumulative reward but rather the probability of receiving the highest cumulative reward.
Note that the proposed equilibrium for finite-time commit games in this section may be different from the Nash equilibrium or the equilibrium presented in Section \ref{section_risk_averse_equilibrium}.

\label{section_commit_games}
Consider that the $N$ players plan to play a game for $M$ independent times where all players have to commit to the pure strategy they play in the first round for the whole game. The strategy in the first round of the game does not have to be pure and can be mixed, but a player has to commit to the randomly sampled pure strategy according to the mixed strategy for $M$ times.
Let $U_i^M(s_i, \boldsymbol{s}_{-i}) = U_i^M(\boldsymbol{s}) = X^1 + X^2 + \dots + X^M$, where $X^j$ for $1 \leq j \leq M$ are independent and identically distributed random variables and $X^1 \sim f_i(x | \boldsymbol{s})$. If players choose to play $\boldsymbol{s} \in \boldsymbol{S}$ for the whole game with $M$ rounds, the random variable $U_i^M(\boldsymbol{s})$ denotes the cumulative payoff for player $i \in [N]$ in the end of the $M$ plays and $U_i^M(\boldsymbol{s}) \sim f_i^M(x | \boldsymbol{s}) = \underbrace{f_i(x | \boldsymbol{s}) \circledast \dots \circledast f_i(x | \boldsymbol{s})}_\text{$M$ times}$.


If the $[N] \setminus i$ players choose to play a mixed strategy $\boldsymbol{\sigma}_{-i}$ in the first round of the game and commit to it for $M-1$ other rounds of the game, using the law of total probability, the distribution of the cumulative payoff for player $i$ in the end of the game when he/she plays $s_i$, denoted by $\overline{U}_i^M(s_i, \boldsymbol{\sigma}_{-i})$, has the probability distribution function
\begin{equation}
    \label{mixed_distribution_M}
    \bar{f}_i^M(x | (s_i, \boldsymbol{\sigma}_{-i})) = \sum_{\boldsymbol{s}_{-i} \in \boldsymbol{S}_{-i}} \Bigg ( f_i^M(x | (s_i, \boldsymbol{s}_{-i})) \cdot \boldsymbol{\sigma}(\boldsymbol{s}_{-i}) \Bigg ),
\end{equation}
where $\boldsymbol{\sigma}(\boldsymbol{s}_{-i}) =  \prod_{j \in [N] \setminus i} \sigma_j (s_j)$ and $s_j$ is the corresponding strategy of player $j$ in $\boldsymbol{s}_{-i}$.
Note that for $s_i, s_i' \in S_i$, the random variables $\overline{U}_i^M(s_i, \boldsymbol{\sigma}_{-i})$ and $\overline{U}_i^M(s_i', \boldsymbol{\sigma}_{-i})$ are not independent from each other in a single play of the game.
The risk-averse equilibrium for an $M$-time commit game can be derived similarly to the derivations in Section \ref{section_risk_averse_equilibrium} and is described below.
From an individual player's point of view, the best response to a mixed strategy of the rest of the players for an $M$-time commit game is defined
as follows.

\begin{definition}
\label{def_best_response_mixed_M_commit}
The set of mixed strategy risk-averse best responses of player $i$ to the mixed strategy profile $\boldsymbol{\sigma}_{-i}$ for an $M$-time commit game is
the set of all probability distributions over the set
\begin{equation}
    \label{eq_mixed_best_response_M_commit}
    \argmax_{s_i \in S_i} P \Big ( \overline{U}_i^M \left ( s_i, \boldsymbol{\sigma}_{-i} \right ) \geq \overline{\boldsymbol{U}}_i^M \left ( S_i \setminus s_i, \boldsymbol{\sigma}_{-i} \right ) \Big ),
\end{equation}
where what we mean by $\overline{U}_i^M \left ( s_i, \boldsymbol{\sigma}_{-i} \right )$ being greater than or equal to $\overline{\boldsymbol{U}}_i^M \left ( S_i \setminus s_i, \boldsymbol{\sigma}_{-i} \right )$ is that  $\overline{U}_i^M \left ( s_i, \boldsymbol{\sigma}_{-i} \right )$ is greater than or equal to $\overline{U}_i^M \left ( s_i', \boldsymbol{\sigma}_{-i} \right )$ for all $s_i' \in S_i \setminus s_i$; otherwise, if $S_i \setminus s_i = \oldemptyset$, player $i$ only has a single option that has to play.
We denote the risk-averse best response set of player $i$'s mixed strategies for an $M$-time commit game, given the other players' mixed strategies $\boldsymbol{\sigma}_{-i}$, by $RB^M(\boldsymbol{\sigma}_{-i})$, which is a set-valued function.
\end{definition}

Given the definition of the risk-averse best response for $M$-time commit games, the risk-averse equilibrium (RAE) for $M$-time commit games is defined
as follows.

\begin{definition}
\label{def_mixed_strategy_RAE_M_commit}
A strategy profile $\boldsymbol{\sigma}^{*M} = (\sigma_1^{*M}, \sigma_2^{*M}, \dots,$ $\sigma_N^{*M})$ is a risk-averse equilibrium (RAE) for an $M$-time commit game, if and only if $\sigma_i^{*M} \in RB^M(\boldsymbol{\sigma}_{-i}^{*M})$ for all $i \in [N]$.
\end{definition}

The following corollary is resulted directly from Theorem \ref{theorem_existence}.

\begin{corollary}
For any finite $N$-player finite-time commit game, a risk-averse equilibrium exists.
\end{corollary}

The pure strategy risk-averse best response for an $M$-time commit game is defined in the following as a specific case of the risk-averse best response defined in Definition \ref{def_best_response_mixed_M_commit}.

\begin{definition}
\label{def_best_response_M}
Pure strategy $\widehat{s}_i$ of player $i$ is a risk-averse best response (RB) to the pure strategy $\boldsymbol{s}_{-i}$ of the other players for an $M$-time commit game if
\begin{equation}
    \label{pure_best_response_M}
    \left\{
    \begin{array}{ll}
        \widehat{s}_i \in \argmax_{s_i \in S_i} P \Big ( U_i^M \left ( s_i, \boldsymbol{s}_{-i} \right ) \geq \boldsymbol{U}_i^M \left ( S_i \setminus s_i, \boldsymbol{s}_{-i} \right ) \Big ), \ \ \ \text{ if } S_i \setminus s_i \neq \oldemptyset,
        \\
        \widehat{s}_i = s_i, \ \ \  \text{ if } S_i \setminus s_i = \oldemptyset,
    \end{array}
    \right.
\end{equation}
where what we mean by $U_i^M \left ( s_i, \boldsymbol{s}_{-i} \right )$ being greater than or equal to $\boldsymbol{U}_i^M \left ( S_i \setminus s_i, \boldsymbol{s}_{-i} \right )$ is that $U_i^M \left ( s_i, \boldsymbol{s}_{-i} \right )$ is greater than or equal to $U_i^M \left ( s_i', \boldsymbol{s}_{-i} \right )$ for all $s_i' \in S_i \setminus s_i$. We denote the risk-averse best response set of player $i$ for an $M$-time commit game, given the other players' pure strategies $\boldsymbol{s}_{-i}$, by $RB^M(\boldsymbol{s}_{-i})$ (overloading notation, $BR^M(.)$ is used for both mixed and pure strategy risk-averse best response for $M$-time commit games).
\end{definition}

Given the definition of the pure strategy risk-averse best response for an $M$-time commit game, the pure strategy risk-averse equilibrium (RAE), which does not necessarily exist, is defined below.

\begin{definition}
\label{def_pure_strategy_RAE_M}
A pure strategy profile $\boldsymbol{s}^{*M} = (s_1^{*M}, s_2^{*M}, \dots,$ $s_N^{*M})$ is a pure strategy risk-averse equilibrium (RAE) for an $M$-time commit game, if and only if $s_i^{*M} \in RB^M(\boldsymbol{s}_{-i}^{*M})$ for all $i \in [N]$.
\end{definition}

\section{Numerical Results}
\label{numerical_results}
In this section, the classical Nash equilibrium is compared with the proposed risk-averse equilibrium. To this end, the likelihood of receiving the higher reward in a two-player game is evaluated under the two types of equilibria for the following example.

\begin{example}
\label{example4}
Consider a game between two players where each player has two pure strategies, $S_1 = \{U, D\}$ and $S_2 = \{L, R\}$, with independent payoff distributions specified as
\begin{enumerate}[label=(\roman*)]
    \item $U_1(U, L)$ and $U_2(U, L)$ are independent and have the same pdf as $f_{1, 1}(u) = \alpha \Big ( 3e^{-20(u - 1)^2} \cdot \mathds{1}\{ \frac12 \leq u \leq \frac32 \} + 2e^{-20(u - a)^2} \cdot \mathds{1}\{ a - \frac12 \leq u \leq a + \frac12 \} \Big )$,
    \item $U_1(U, R)$, $U_2(U, R)$, $U_1(D, L)$, and $U_2(D, L)$ are independent and have the same pdf as $f_{1, 2}(u) = \beta e^{-20(u - 3)^2} \cdot \mathds{1}\{ \frac52 \leq u \leq \frac72 \}$,
    \item $U_1(D, R)$ and $U_2(D, R)$ are independent and have the same pdf as $f_{2, 2}(u) = \gamma \Big ( 7e^{-20(u - 2)^2} \cdot \mathds{1}\{ \frac32 \leq u \leq \frac52 \} + 3e^{-20(u - a - 2)^2} \cdot \mathds{1}\{ a + \frac32 \leq u \leq a + \frac52 \} \Big )$,
\end{enumerate}
where $\alpha, \beta$, and $\gamma$ are constants for which each of the corresponding distributions integrate to one, $a \geq 0$ is a constant, and $\mathds{1}\{.\}$ is the indicator function.

\end{example}

The Nash equilibrium in the above example depends on the value of the constant $a$. If $0 \leq a < 3.333$, pure strategies $(U, R)$ and $(D, L)$ are the two pure Nash equilibria of the game in addition to a mixed Nash equilibrium. If $3.333 \leq a \leq 6$, the pure strategy $(D, R)$ is the only Nash equilibrium of the game. If $a > 6$, pure strategies $(U, L)$ and $(D, R)$ are the two pure Nash equilibria of the game in addition to a mixed Nash equilibrium.
On the other hand, no matter what the value of the constant $a$ is, the game has a mixed risk-averse equilibrium as well as two pure risk-averse equilibria, which are $(U, R)$ and $(D, L)$.
The mixed strategy Nash equilibrium and the mixed strategy risk-averse equilibrium depend on the value of the constant $a$. Note that the game is symmetric from the perspective of the two players, so the mixed strategy Nash and risk-averse equilibria are characterized by $\sigma_1(U)$ that is plotted in Figure \ref{figure_example4}.

\begin{figure}[t]
\centering
\includegraphics[width=0.7\textwidth]{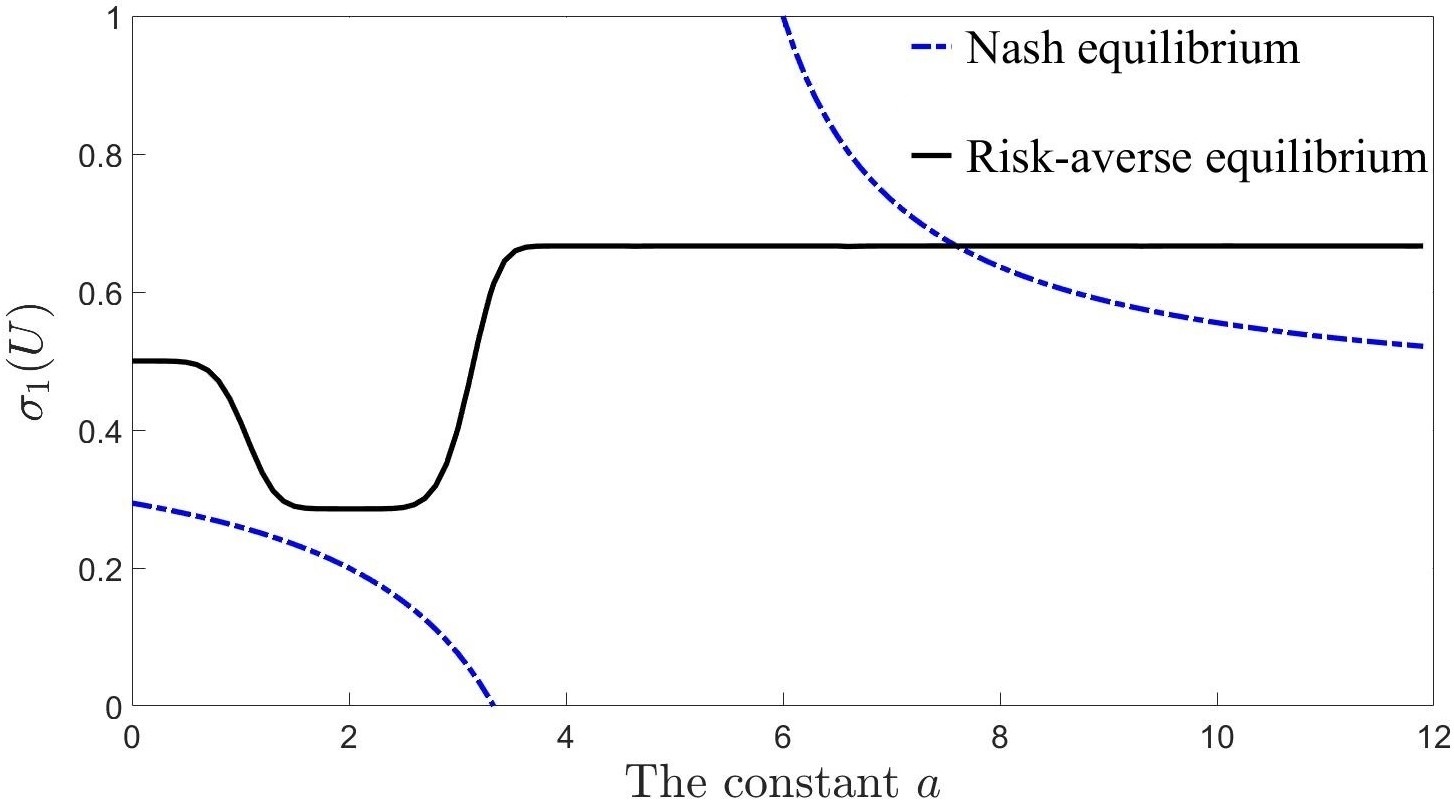}
\caption{The mixed strategy Nash and risk-averse equilibria are determined by the value of $\sigma_1(U)$ in Example \ref{example4}. The mixed strategies depend on the value of the constant $a$, where $\sigma_1(U)$ is plotted above as a function of the constant $a$.}
\label{figure_example4}
\end{figure}

A game according to Example \ref{example4} is simulated for $10^6$ rounds for a fixed constant $a$. In each realization of the game, both Nash and risk-averse equilibria are played and their corresponding payoffs are compared for one of the players to see which one is larger. The mixed strategies under Nash and risk-averse equilibria are compared against each other and the pure strategies under Nash and risk-averse equilibria, if different, are compared against each other.
After the $10^6$ games, the proportion of the games in which playing according to the risk-averse equilibrium outperforms playing according to the Nash equilibrium by having a larger payoff is computed and
plotted in Figure \ref{result} as a function of the constant $a$.
In Figure \ref{result}, the plots for comparison of the Nash and risk-averse mixed strategies are dotted lines, the plot for comparison of the Nash equilibrium $(D, R)$ and risk-averse pure strategy $(U, R)$ or $(D, L)$ is a solid line for $a > 3.333$, and the plot for comparison of the Nash equilibrium $(U, L)$ and risk-averse pure strategy $(U, R)$ or $(D, L)$ is a dash-dotted line for $a > 3.333$. Note that the payoff distributions are the same for pure risk-averse equilibria (PRAE) $(U, R)$ and $(D, L)$, that is why pure Nash equilibria $(D, R)$ (PNE1) and $(U, L)$ (PNE2) are compared with one of the PRAEs.
For example, the mixed risk-averse equilibrium is compared against the mixed Nash equilibrium for $0 \leq a < 3.333$, but the pure strategies under both approaches are the same in this interval; that is why the pure strategies are not compared for $0 \leq a < 3.333$.
As shown in the figure, it is more likely to receive a higher payoff under the risk-averse equilibrium than under the Nash equilibrium in a single play of the game.

\begin{figure}[t]
\centering
\includegraphics[width=0.7\textwidth]{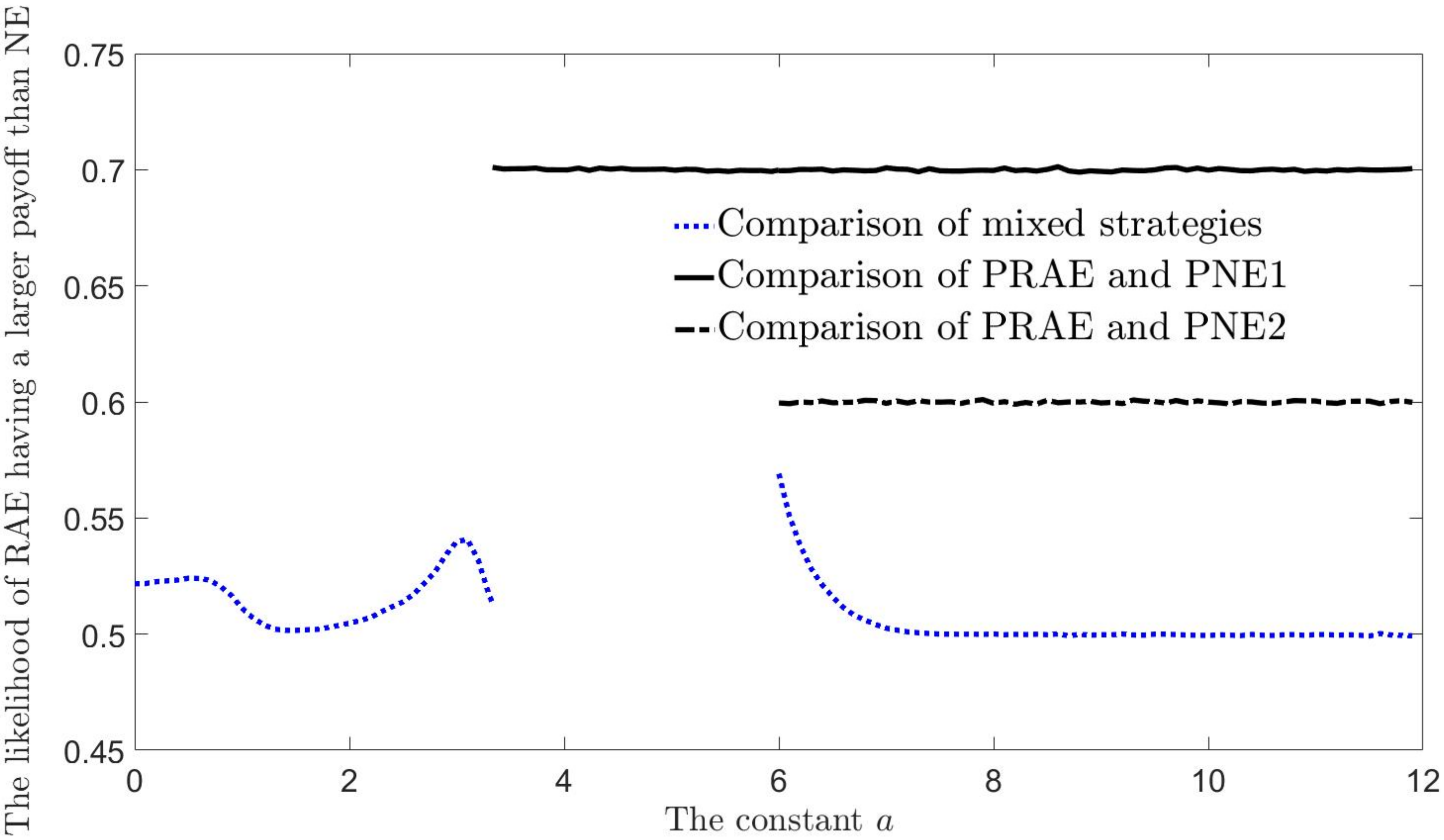}
\caption{The likelihood of the payoff of the risk-averse equilibrium being greater than the payoff of the Nash equilibrium.}
\label{result}
\end{figure}

\section{Conclusion and Future Work}
\label{conclusion_future}

We have proposed the risk-averse equilibrium (RAE) for stochastic games as a method for examining one-shot or limited-run games with stochastic payoffs.
In such a setting, it makes sense to consider players who want to maximize not the expected payoff but rather the likelihood of receiving the largest payoff.
In doing so, we draw parallels to prospect theory in economic decisions, where consumers prefer an option with lower variance at the cost of lower expected utility, rather than an option with higher expected utility at the cost of higher variance when facing significant decisions.
We then propose the risk-averse equilibrium to address one-shot games in such a situation and show it to exist in any $N$-player finite stochastic game.
We prove the existence of the risk-averse equilibrium independent of Nash equilibrium along with familiar concepts such as strategy dominance. We also define a probability tensor and show that the risk-averse equilibria of a game are equivalent to the Nash equilibria of such a tensor.
We next considered the risk-averse equilibrium in limited-run games by examining $M$-time commit
games,
where players commit to a strategy for the $M$ rounds of the stochastic game.


Looking forward, the risk-averse equilibrium allows competition to be incorporated into many traditional risk-averse settings. Election modeling is one such example with a limited-run of interactions between candidates, and each candidate wants to maximize not the expected votes they receive, but rather their probability of winning. This is one of the drawbacks to the widely used Hotelling-Downs model (\cite{downs1957}), which assumes candidates merely want to maximize their expected voter share. By instead maximizing their probability of making the best response to other candidates, the risk-averse equilibrium will be able to address this shortcoming while offering similar interpetability and insight.

\bibliographystyle{informs2014trsc}
\bibliography{myrefs}


%
%
%




\begin{appendices}
\label{appendix}
The risk-averse best response of player $i$ to the strategy profile $\boldsymbol{\sigma}_{-i}$ is presented in Definition \ref{def_best_response_mixed} as the set of all probability distributions over the set
\begin{equation}
    \argmax_{s_i \in S_i} P \Big ( \overline{U}_i \left ( s_i, \boldsymbol{\sigma}_{-i} \right ) \geq \overline{\boldsymbol{U}}_i \left ( S_i \setminus s_i, \boldsymbol{\sigma}_{-i} \right ) \Big ).
\end{equation}
The same randomness on the action of players $[N] \setminus i$ is considered in $\overline{U}_i(s_i, \boldsymbol{\sigma}_{-i})$ for all $s_i \in S_i$ in this article. That is why for $s_i \neq s_i' \in S_i$, the random variables $\overline{U}_i(s_i, \boldsymbol{\sigma}_{-i})$ and $\overline{U}_i(s_i', \boldsymbol{\sigma}_{-i})$ are not independent of each other in a single play of the game.
On the other hand, independent randomness on the action of players $[N] \setminus i$ can be considered in $\overline{U}_i(s_i, \boldsymbol{\sigma}_{-i})$ for all $s_i \in S_i$.
In that case, $\overline{U}_i(s_i, \boldsymbol{\sigma}_{-i})$ is independent from $\overline{U}_i(s_i', \boldsymbol{\sigma}_{-i})$ for all $s_i \neq s_i' \in S_i$, so
\begin{equation}
\begin{aligned}
    \label{eq_proof_strict_dominance2}
    & P \Big ( \overline{U}_i \left ( s_i, \boldsymbol{\sigma}_{-i} \right ) \geq \overline{\boldsymbol{U}}_i \left ( S_i \setminus s_i, \boldsymbol{\sigma}_{-i} \right ) \Big ) \\
     \overset{(a)}{=} & \int \dots \int_{x_{s_i} \geq \boldsymbol{x}_{S_i \setminus s_i}} \left ( \prod_{{s}_i' \in S_i} \bar{f}_i(x_{{s}_i'} | ({s}_i', \boldsymbol{\sigma}_{-i})) dx_{{s}_i'} \right ) \\
    \overset{(b)}{=} & \int \dots \int_{x_{s_i} \geq \boldsymbol{x}_{S_i \setminus s_i}} \left ( \prod_{{s}_i' \in S_i} \sum_{\boldsymbol{s}_{-i} \in \boldsymbol{S}_{-i}} \Bigg ( f_i(x_{{s}_i'} | ({s}_i', \boldsymbol{s}_{-i})) \cdot \boldsymbol{\sigma}(\boldsymbol{s}_{-i}) \Bigg ) dx_{{s}_i'} \right ) \\
    \overset{(c)}{=} & \sum_{s_{-i}^{1:|S_i|} \in \boldsymbol{S}_{-i}^{|S_i|}} \left ( \left ( \prod_{k = 1}^{|S_i|} \boldsymbol{\sigma}(\boldsymbol{s}_{-i}^k) \right ) \cdot P \Big (U_i(s_i, \boldsymbol{s}_{-i}^1) \geq \boldsymbol{U}_i(S_i \setminus s_i, \boldsymbol{s}_{-i}^{2:|S_i|}) \Big )  \right ),
\end{aligned}
\end{equation}
where $(a)$ follows by the fact that all payoff distributions are independent of each other, so the pdf functions can be multiplied together to get the joint distribution of $\overline{U}_i(s_i, \boldsymbol{\sigma}_{-i})$ for all $s_i \in S_i$, $(b)$ follows by Equation \eqref{mixed_distribution}, $(c)$ is true by expanding the product and reformulating the product of the sum as the sum of products.
If Equation \eqref{eq_proof_strict_dominance2} is used in Equation \eqref{eq_find_mixed_strategy_RAE} to find the equilibrium of the game, we come up with a different equilibrium than that presented in the main body of this article. Let the equilibrium derived from Equations \eqref{eq_find_mixed_strategy_RAE} and \eqref{eq_proof_strict_dominance2} be called RAE$_2$, where following the same proof of Theorem \ref{theorem_existence}, RAE$_2$ exists for any finite $N$-player game.
Finding a strictly dominated strategy in the framework of RAE$_2$ is not as straightforward as for the Nash and RAE equilibria. In the following definition, the strict dominance is described for RAE$_2$.

\begin{definition}
\label{def_strict_dominance2}
A pure strategy $s_i \in S_i$ of player $i$ strictly dominates a second pure strategy $s_i' \in S_i$ of the player if
\begin{equation}
    \label{eq_strict_dominance2}
    P \Big ( U_i \left ( s_i, \boldsymbol{s}_{-i}^1 \right ) \geq \boldsymbol{U}_i \left ( S_i \setminus s_i, \boldsymbol{s}_{-i}^{2:|S_i|} \right ) \Big ) > P \Big ( U_i \left ( s_i', \boldsymbol{s}_{-i}^1 \right ) \geq \boldsymbol{U}_i \left ( S_i \setminus s_i', \boldsymbol{s}_{-i}^{2:|S_i|} \right ) \Big ), \forall \boldsymbol{s}_{-i}^{1:|S_i|} \in \boldsymbol{S}_{-i}^{|S_i|},
\end{equation}
where what we mean by $U_i \left ( s_i, \boldsymbol{s}_{-i}^1 \right )$ being greater than or equal to $\boldsymbol{U}_i \left ( S_i \setminus s_i, \boldsymbol{s}_{-i}^{2:|S_i|} \right )$ is that $U_i \left ( s_i, \boldsymbol{s}_{-i}^1 \right )$ is greater than or equal to $U_i \left ( \widehat{s}_i, \boldsymbol{s}_{-i}^k \right )$ for all $\widehat{s}_i \in S_i \setminus s_i$, where each $\widehat{s}_i \in S_i \setminus s_i$ is associated with a possibly different pure strategy of other players $\boldsymbol{s}_{-i}^k \in \boldsymbol{S}_{-i}$ for all $2 \leq k \leq |S_i|$. Note that the associations of $\widehat{s}_i \in S_i$ and $\boldsymbol{s}_{-i}^k \in \boldsymbol{S}_{-i}$ on both sides of Equation \eqref{eq_strict_dominance2} remain the same except for $s_i$ and $s_i'$ for which the associations are switched with each other.
\end{definition}

Note that the strictly dominated strategies of a player cannot be found from the risk-averse probably matrix, but finding a strictly dominated strategy needs more sophisticated calculations described in Definition \ref{def_strict_dominance2}.
A strictly dominated strategy cannot be the risk-averse best response to any mixed strategy profile of other players due to the following reason.
Consider that $s_i' \in S_i$ is strictly dominated by $s_i \in S_i$ for player $i$ as is stated in Definition \ref{def_strict_dominance2}. Then, for any $\boldsymbol{\sigma}_{-i} \in \boldsymbol{\Sigma}_{-i}$, we have
\begin{equation}
\begin{aligned}
    \label{eq_proof_strict_dominance22}
    & P \Big ( \overline{U}_i \left ( s_i, \boldsymbol{\sigma}_{-i} \right ) \geq \overline{\boldsymbol{U}}_i \left ( S_i \setminus s_i, \boldsymbol{\sigma}_{-i} \right ) \Big ) \\
     \overset{(a)}{=} & \sum_{s_{-i}^{1:|S_i|} \in \boldsymbol{S}_{-i}^{|S_i|}} \left ( \left ( \prod_{k = 1}^{|S_i|} \boldsymbol{\sigma}(\boldsymbol{s}_{-i}^k) \right ) \cdot P \Big (U_i(s_i, \boldsymbol{s}_{-i}^1) \geq \boldsymbol{U}_i(S_i \setminus s_i, \boldsymbol{s}_{-i}^{2:|S_i|}) \Big )  \right ) \\
    \overset{(b)}{>} & \sum_{s_{-i}^{1:|S_i|} \in \boldsymbol{S}_{-i}^{|S_i|}} \left ( \left ( \prod_{k = 1}^{|S_i|} \boldsymbol{\sigma}(\boldsymbol{s}_{-i}^k) \right ) \cdot P \Big (U_i(s_i', \boldsymbol{s}_{-i}^1) \geq \boldsymbol{U}_i(S_i \setminus s_i', \boldsymbol{s}_{-i}^{2:|S_i|}) \Big )  \right ) \\
    \overset{(c)}{=} & P \Big ( \overline{U}_i \left ( s_i', \boldsymbol{\sigma}_{-i} \right ) \geq \overline{\boldsymbol{U}}_i \left ( S_i \setminus s_i', \boldsymbol{\sigma}_{-i} \right ) \Big ),
\end{aligned}
\end{equation}
where $(a)$ is true by Equation \eqref{eq_proof_strict_dominance2}, $(b)$ is true by the assumption that the pure strategy $s_i'$ is strictly dominated by the pure strategy $s_i$ and using Equation \eqref{eq_strict_dominance2} in Definition \ref{def_strict_dominance2} on strict dominance, and $(c)$ follows the backward direction of steps $(a), (b),$ and $(c)$ for pure strategy $s_i'$ in Equation \eqref{eq_proof_strict_dominance2}.
By Equation \eqref{eq_proof_strict_dominance22} and Equation \eqref{eq_mixed_best_response} in Definition \ref{def_best_response_mixed} on the best response to a mixed strategy profile of other players, a strictly dominated pure strategy can never be a best response to any mixed strategy profile of other players. As a result, a strictly dominated pure strategy can be removed from the set of strategies of a player, and iterated elimination of strictly dominated strategies can be applied to a game under the framework of RAE$_2$ as well.

In order to get more insight into the new framework, the mixed strategy RAE$_2$ is worked out for Example \ref{example2}.
Consider that the first player selects $U$ with probability $\sigma_U$ and selects $D$ otherwise. Given the first player's mixed strategy $(\sigma_U, 1 - \sigma_U)$, with a little misuse of notation, denote the random variables denoting the second player's payoffs by selecting $L$ or $R$ with $L$ and $R$, respectively. As a result, for two independent games, where in both of them the first player independently plays according to the mixed strategy $(\sigma_U, 1- \sigma_U)$ and the second player selects $L$ and $R$ in the first and second games, respectively, using the law of total probability,
\begin{equation}
\label{eq_mixed_RAE_example22}
    \begin{aligned}
        & L \sim f_L(u) = \sigma_U \cdot f_4(u) + (1 - \sigma_U) \cdot f_3(u),\\
        & R \sim f_R(v) = \sigma_U \cdot f_3(v) + (1 - \sigma_U) \cdot f_1(v).
    \end{aligned}
\end{equation}
The second player is indifferent between selecting $L$ and $R$ if $P(L \geq R) = P(R \geq L)$. Since payoffs are continuous random variables, $P(R \geq L) = 1 - P(L \geq R)$; as a result, the second player is indifferent between the strategies in two independent games if $P(L \geq R) = 0.5$. By Equation \eqref{eq_mixed_RAE_example22} and the fact that payoffs are independent from each other, $P(L \geq R)$ can be computed as
\begin{equation}
    \begin{aligned}
        P(L \geq R) & = \int_\infty^\infty \int_v^\infty f_L(u) \cdot f_R(v) \ dudv \\
        & = \int_\infty^\infty \int_v^\infty \Big ( \sigma_U \cdot f_4(u) + (1 - \sigma_U) \cdot f_3(u) \Big ) \cdot \Big ( \sigma_U \cdot f_3(v) + (1 - \sigma_U) \cdot f_1(v) \Big ) \ dudv \\
        & = \sigma_U^2 \int_\infty^\infty \int_v^\infty f_4(u) \cdot f_3(v) \ dudv + \sigma_U(1 - \sigma_U) \int_\infty^\infty \int_v^\infty f_4(u) \cdot f_1(v) \ dudv + \\
        & \hspace{4.5mm} \sigma_U(1 - \sigma_U) \int_\infty^\infty \int_v^\infty f_3(u) \cdot f_3(v) \ dudv + (1 - \sigma_U)^2 \int_\infty^\infty \int_v^\infty f_3(u) \cdot f_1(v) \ dudv \\
        & = \sigma_U^2 P \big ( U_2(U, L) \geq U_2(U, R) \big ) + \sigma_U(1 - \sigma_U) P \big ( U_2(U, L) \geq U_2(D, R) \big ) + \\
        & \hspace{4.5mm} \sigma_U(1 - \sigma_U) P \big ( U_2(D, L) \geq U_2(U, R) \big ) + (1 - \sigma_U)^2 P \big ( U_2(D, L) \geq U_2(D, R) \big ) \\
        & = \frac25\sigma_U^2 + \sigma_U(1 - \sigma_U) + \frac12\sigma_U(1 - \sigma_U) + (1 - \sigma_U)^2 = -\frac{1}{10}\sigma_U^2 - \frac12\sigma_U + 1.
    \end{aligned}
\end{equation}
Letting $P(L \geq R) = 0.5$, then $-\frac{1}{10}\sigma_U^2 - \frac12\sigma_U + \frac12 = 0$ whose solution is the mixed strategy RAE$_2$. It can be computed that $\sigma_U = \frac{-5 + \sqrt{45}}{2} \approx 0.854$. As a result, due to symmetry, $\big (\sigma_1(U), \sigma_1(D) \big ) = (0.854, 0.146)$ and $\big (\sigma_2(L), \sigma_2(R) \big ) = (0.854, 0.146)$ form the mixed strategy RAE$_2$ of the game in Example \ref{example2}.

The risk-averse best response under the RAE$_2$ framework is compared against the risk-averse best response under the RAE framework by simulation for Example \ref{example4}.
The mixed strategy RAE and RAE$_2$ exist no matter what the value of the positive constant $a$ is in this example, but these equilibria depend on the value of the constant $a$.
As described in Section \ref{numerical_results}, the mixed strategy RAE and RAE$_2$ are characterized by $\sigma_1(U)$ that is plotted in Figure \ref{figure_example4Appendix}.

\begin{figure}[t]
\centering
\includegraphics[width=0.7\textwidth]{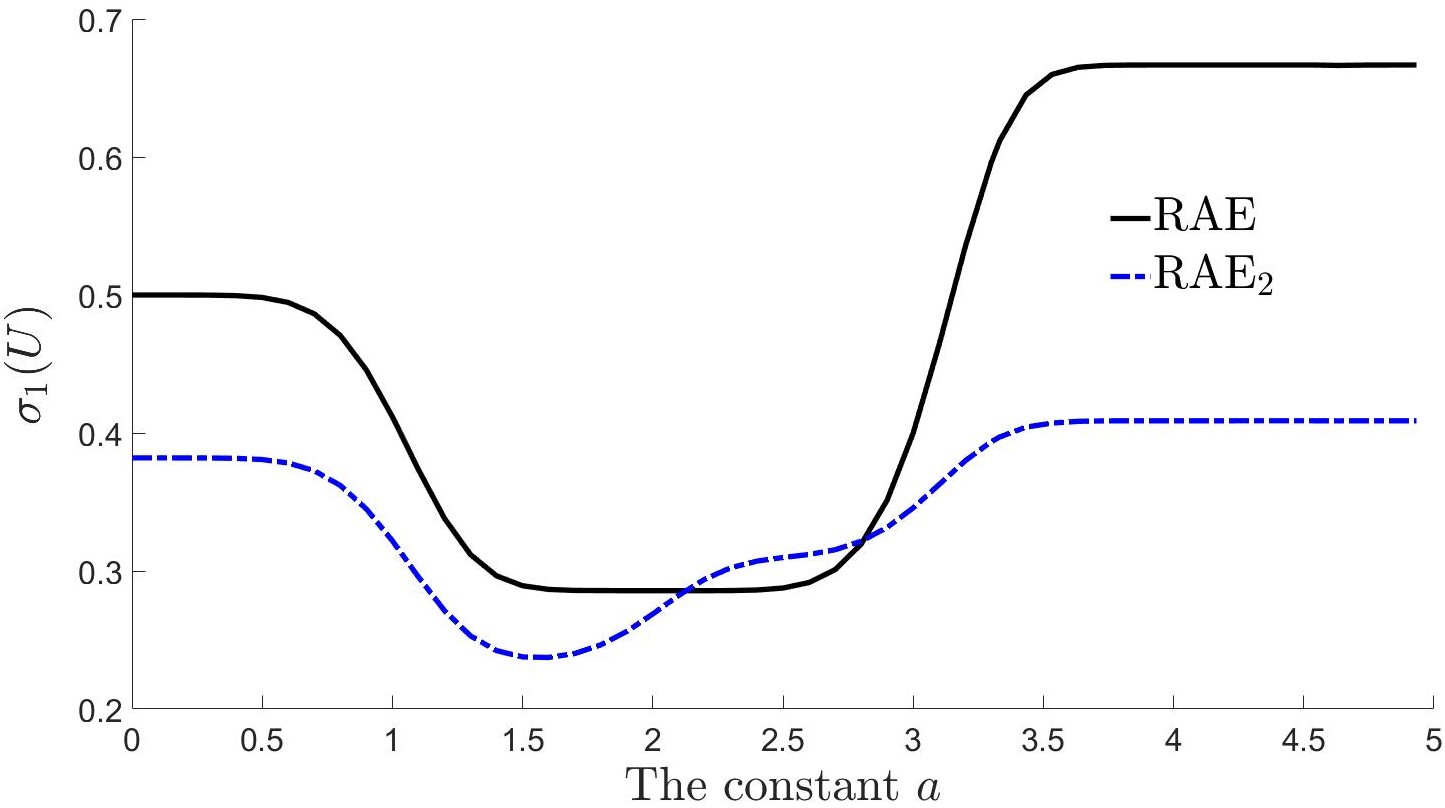}
\caption{The mixed strategy RAE and RAE$_2$ are determined by the value of $\sigma_1(U)$ in Example \ref{example4}. The mixed strategies depend on the value of the constant $a$, where $\sigma_1(U)$ is plotted above as a function of the constant $a$.}
\label{figure_example4Appendix}
\end{figure}

\begin{figure}[t]
\centering
\includegraphics[width=0.7\textwidth]{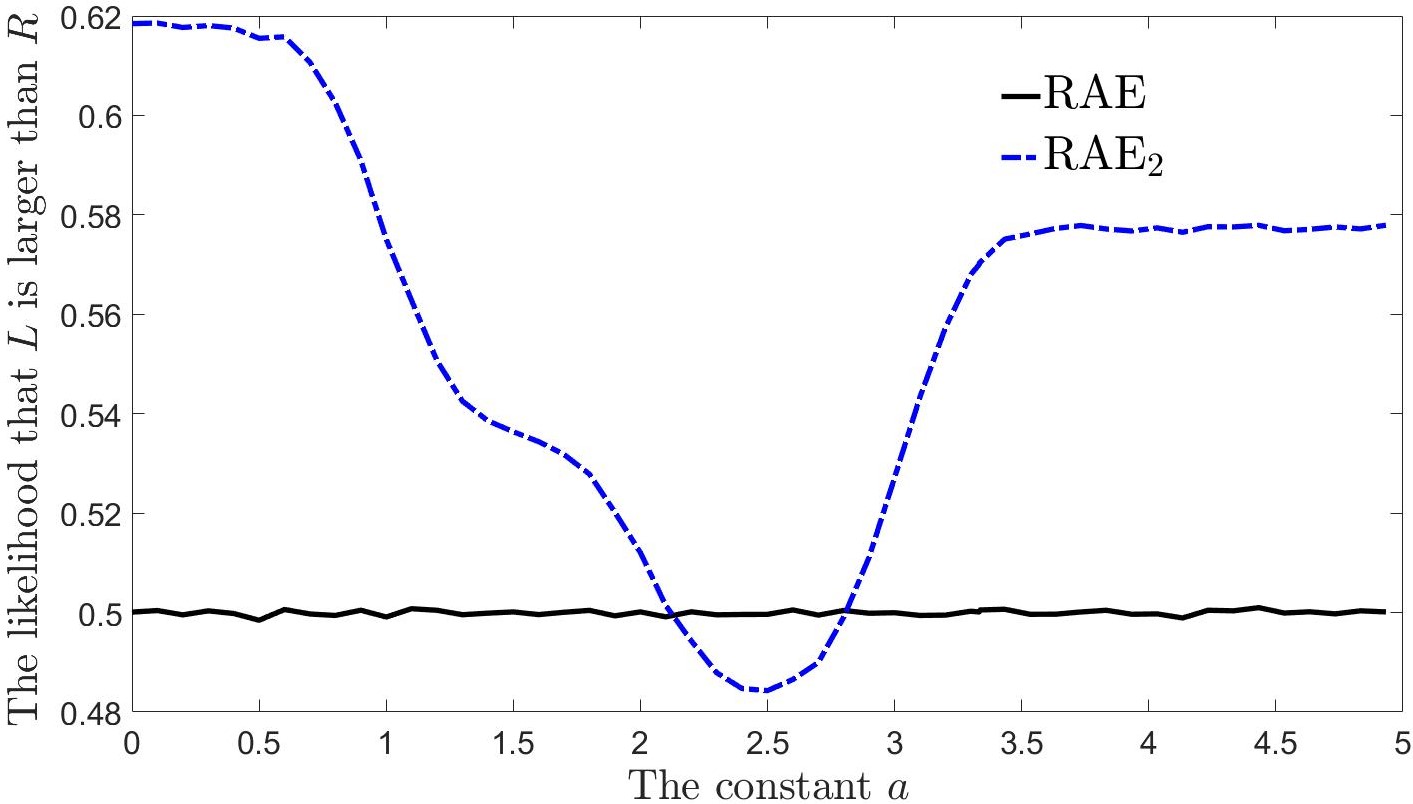}
\caption{The likelihood that playing strategy $L$ outperforms playing strategy $R$ by having a larger payoff in a single play of the game.}
\label{LgeqR}
\end{figure}

A game according to Example \ref{example4} is simulated for $10^6$ rounds for a fixed constant $a$. In each realization of the game, the first player selects a strategy according to the mixed strategy RAE, then the payoffs of the second player for his two strategies are compared to see which is larger.
After the $10^6$ games, the proportion of the games in which playing strategy $L$ outperforms playing strategy $R$ by having a larger payoff is computed and
plotted in Figure \ref{LgeqR} as a function of the constant $a$. The same procedure is performed for the mixed strategy RAE$_2$ and the result is plotted in the same figure.
As shown in the figure, under the RAE framework, the likelihood that playing strategy $L$ has a larger payoff than playing strategy $R$ is the same as that of observing heads on the flip of a fair coin, which makes the second player indifferent between his two strategies.
However, if the first player selects her strategy according to the mixed strategy RAE$_2$, it is more likely for the second player to get a larger payoff by selecting $L$ or $R$ except for two specific values of the constant $a$ as shown in the figure.
As a result, the risk-averse equilibrium presented in the main body of this article makes the second player indifferent between his two strategies since the chances of receiving a larger payoff from either strategy are the same in a single play of the game.

\end{appendices}

\end{document}